\documentclass{article} 
\usepackage{amsmath,amssymb,latexsym}
\usepackage[standard]{ntheorem}
\newcommand{\qed}{\hfill$\Box$}
\usepackage[all]{xy}
\newcommand{\automa}[1]{\par\addvspace{\bigskipamount}%
   \makebox[\linewidth][c]{\xymatrix@=12mm{#1}}\bigskip\\}
\newcommand{\stato}[1]{*++[o][F]{#1}}

\newcommand{\iniziale}[1]{{\ar@{<-}@/#1{1.8em}/[]}}
\newcommand{\arco}[4]{\ar@/#1{#4mm}/#1{#2}[#3]}

\newcommand{\arcoAR}[3]{\ar@<2pt>[#3]^{#1}\ar@{<-}@<-2pt>[#3]_{#2}}
\newcommand{\ciclo}[2]{\ifcase#2 \ar@(lu,u)^{#1}[]\or%
  \ar@(u,ur)^{#1}[]\or\ar@(ur,r)^{#1}[]\or\ar@(r,rd)^{#1}[]\or\ar@(rd,d)^{#1}[]\or%
  \ar@(d,dl)^{#1}[]\or\ar@(dl,l)^{#1}[]\or\ar@(l,lu)^{#1}[]\or\ar@(lu,u)^{#1}[]\fi}

\newcommand{\fa}{\varphi_\mathcal A}
\newcommand{\A}{\mathcal A}
 
\newcommand{\qedafter}[1][2\bigskipamount]%
{\unskip\nobreak\hfill\nobreak\raisebox{-#1}[9pt][0pt]{\ensuremath{\Box}}}

\newcommand{\card}{\mathop{\mathrm{Card}}\nolimits}
\newcommand{\rk}{\mathop{\mathrm{rank}}\nolimits}

\newcommand{\N}{\mathbb{N}}

\newcommand{\Q}{\mathbb{Q}}

\title{Independent sets of words and the synchronization problem
\thanks{This work was partially supported by MIUR project
{``Aspetti matematici e applicazioni emergenti degli automi e  dei linguaggi formali'' and by
fundings ``Facolt\`a di 
Scienze MM. FF. NN. 2008'' of the University of Rome ``La Sapienza''.}}}

\author{Arturo Carpi \\
Dipartimento di Matematica e Informatica, \\
 Universit\`a  degli Studi di Perugia, \\
  via Vanvitelli 1,  06123 Perugia, Italy. \\
  e-mail:  carpi@dipmat.unipg.it
 \and Flavio D'Alessandro \\
        Dipartimento di Matematica, \\
Universit\`a di Roma ``La Sapienza'' \\
  Piazzale Aldo Moro 2, 00185 Roma,  Italy. \\
  e-mail: dalessan@mat.uniroma1.it
  \and
  and \\
Department of Mathematics, 
Bo\u gazi\c ci University \\
  34342 Bebek, Istanbul, Turkey.
}

\begin{document}

\maketitle

\begin{abstract}
The synchronization problem is investigated for the class of locally strongly transitive  automata  introduced in \cite{DLT}.
Some extensions of this problem related to the notions of {\em stable set} and  word of  {\em minimal rank} of an automaton are studied.  
An application to synchronizing colorings of aperiodic graphs with a Hamiltonian path is also considered.

\end{abstract}

\medskip

   \section{Introduction}
     A deterministic   automaton is called {\em synchronizing} if there exists an
   input-sequence, called {\em synchronizing} or {\em reset word}, such that the state attained
by the automaton, when this sequence is read, does not depend on the initial state of the automaton
itself. Two fundamental problems which have been intensively investigated in the last  decades
are based upon this concept: the  {\em \v Cern\'y conjecture} and the {\em Road coloring problem}.
 
 The  \v Cern\'y conjecture  \cite{Cerny} claims that  a deterministic synchronizing $n$-state automaton has a reset word of length  
$(n-1)^2$.  This conjecture and some related problems have been widely investigated  in several papers  ({\em cf.}
 \cite{A, AGV,Beal,BP,Acta,DLT,Cerny,Dubuc,Frankl,Kari,Pin78,Pin78bis,Rys,Traht}). 
The interested reader is refered to \cite{V08} for a historical survey of the  \v Cern\'y conjecture 
and to  \cite{Carpi} for synchronizing unambiguous automata.
 
In  \cite{DLT}, the authors have introduced   the notion of {\em local strong transitivity}.
An $n$-state automaton $\A$  is said to be {\em locally strongly transitive} if 
it is equipped by  a set  $W$ of $k$ words and a set  $R$  of $k$ distinct states
such that, for all states $s$ of $\A$ 
and all $r\in R$, there exists a word $w\in W$  taking the state $s$ into $r$. 
 	The set $W$ is  called  {\em independent} while $R$ is called the  {\em range} of $W.$
The main result of  \cite{DLT} is that 
 any synchronizing locally strongly transitive $n$-state automaton has a 
reset word of length not larger than
$(k-1)(n + L) + \ell,$
where $k$ is the cardinality of an independent set $W$ and $L$ and $\ell$ denote  
respectively the maximal and the minimal  length of the 
words of $W$.

In the case where all the states of the automaton are in the range, the automaton $\A$ is said to be
{\em strongly transitive}. {Strongly transitive automata} have been studied in \cite{Acta}. This notion is related 
with that of regular automata introduced in \cite{Rys}.
 	A remarkable example of locally strongly transitive automata is that of {\em 1-cluster automata} introduced in \cite{BP}. 
	An automaton   is called 1-cluster
	if  there exists  a letter $a$ such that the graph of the automaton has a unique cycle  labelled by a power of $a$. 
	One can easily verify that a $n$-state automaton is 1-cluster if and only if it has an independent set of words of the form
	 
$$
	 \{a^{n-1},a^{n-2},\ldots, a^{n-k}\}.
$$
Moreover one can take $k$ equal to the length of the unique cycle  labelled by a power of $a$. 
  
   In this paper, by developing the techniques of \cite{DLT} and \cite{DLT10} on locally strongly transitive automata, 
 we investigate  the synchronization problem and some related topics.
A remarkable result we prove, shows that 
 any synchronizing locally strongly transitive $n$-state automaton has a 
reset  word of length not larger than
$$ (k -1)(n+L+1) - 2k\ln \frac{k+1}{2}+\ell,$$
where $k$ is the cardinality of an independent set $W$ and $L$ and $\ell$ 
denote  respectively the maximal and the minimal  length of the 
words of $W$.
As a straightforward corollary of this result,  we prove that every $n$-state $1$-cluster 
synchronizing  automaton has a  reset word
of length not larger than  $$2n^2 -4n +1 -2(n-1) \ln\frac{n}{2},$$
so  recovering, for such automata,  some results of B\'eal et al. \cite{BBP} and Steinberg \cite{Stein} 
with an improved bound.

We further investigate two notions that are strongly related with some extensions of the synchronization problem: 
the notion of {\em stable set} and that of word of  {\em minimal rank} of an automaton.  
Given an automaton $\A=\langle Q,A,\delta\rangle$, a set $K$ of states  of $\A$ is 
{\em reducible} if there exists a word $w\in A^*$ taking all the states
of $K$ into a fixed state. A set $K\subseteq Q$ is \emph{stable} if for any $p,q\in K$, and for any $w\in A^*$,
the set $\{\delta(p, w), \delta (q, w)\}$ is reducible. The concept of stability was introduced in \cite{ckk} and plays a fundamental role in the solution  \cite{trah}  of the Road coloring problem. Clearly if $\A$ is synchronizing, then every subset  of $Q$ is
stable. Thus a question that naturally arises in this context is to evaluate, for a given stable subset $K$ 
in a non-synchronizing automaton, 
the minimal length of a word $w$ such that $\card(\delta (K, w))=1$. 
We prove that if $\A$ is a locally strongly
transitive $n$-state automaton,  then  
the minimal length of such a word $w$ 
 is at most  
\begin{equation}\label{eq:int}
(M-1)(n+L+1)-k \ln M + L,
\end{equation}
where $k$ is the cardinality of any  independent set $W$, $L$   denotes  
 the maximal   length of the 
words of $W$, and $M$ is the maximal cardinality of reducible subsets of the range of $W$.

The second topic that we investigate concerns the construction  of
words of minimal rank of an automaton.   
 The  {rank} of  a word $w$ in an automaton $\A$  is the  cardinality of 
 the set of states $\delta (Q, w)$. Clearly $w$ is a reset word if and only if its rank is $1$.   	
 The length of words of minimal rank  in an automaton was first investigated by Pin in \cite{Pin78, Pin78bis} 
 for deterministic automata and by Carpi in   \cite{Carpi} for unambiguous automata.
 In this context, we prove that, if $\A$ is a locally strongly transitive automaton, 
 and $t$ is the minimal rank  of its words, 
 then there exists a word $u$ of  rank $t$ and  length  
  $$
	|u|\leq	 \ell +(k-t)(L +n+1)-tk\ln\frac kt,
	$$  
	where, as before,  $k$ is the cardinality of an independent set $W$ and $L$ and $\ell$ 
denote  respectively the maximal and the minimal  length of the 
words of $W$.  It is also proved that the maximal cardinality of reducible subsets
of the range of $W$ is $M=k/t$
so that (\ref{eq:int}) can be written as 
	\[
	 \left( \frac{k}{t} - 1\right)(n+L +1) - k\ln \frac{k}{t} + L  \,.
	\]
 
  In the case of 1-cluster $n$-state automata, the previous bound becomes 
  	 	\[
	 \frac{2nk}{t} - n - 1  - k\ln \frac{k}{t}.
	\]
  Finally another application of our techniques  concerns the study of a conjecture related to
 the well-known {\em Road coloring problem}.
  This problem asks to determine whether any aperiodic and strongly connected digraph,
  with all vertices of the same outdegree ({\em AGW-graph}, for short)  
   has
   a  {\em synchronizing coloring}, that is, a labeling of its edges  that turns it into a synchronizing
   deterministic automaton.
 The problem was formulated in the context
 of Symbolic Dynamics by Adler, Goodwyn and  Weiss and it is explicitly stated in \cite{AGW}. In 2007, 
 Trahtman  \cite{trah} has positively solved it.  The solution by Trahtman has electrified the community of
 formal language theorists and recently Volkov has raised in  \cite{V} (see also \cite{AGV})
 the problem of evaluating, for any AGW-graph $G$,  
 the minimal length of a reset word for a 
synchronizing coloring of $G$. 
This problem  has been called {\em the Hybrid \v Cern\'y--Road coloring problem.}
 It is worth to mention that Ananichev has found, for any $n\geq 2$,  an AGW-graph of $n$ vertices such that
the length of the shortest reset word for any synchronizing coloring of the graph is  $(n-1)(n-2) + 1$ (see  \cite{AGV}).
In \cite{DLT}, the authors have proven that, 
 given  an AGW-graph $G$ of $n$ vertices, without multiple edges, such that $G$ has a simple cycle   of
prime length $p < n$,   there exists a synchronizing coloring of $G$ with a reset word of length $(2p-1)(n-1)$. 
Moreover, in the case $p=2$, that is, if $G$ contains a cycle of length $2$, 
then, also in presence of multiple edges,
there exists  a synchronizing coloring with a reset word of length  $5(n-1)$.
 
 In this paper, we continue the investigation of     	 
 {the Hybrid \v Cern\'y--Road coloring problem} on a very natural class of digraphs, those having a {\em
 Hamiltonian path}. The main result of this paper states that any AGW-graph $G$ of $n$ vertices with a Hamiltonian
 path  admits a synchronizing  coloring with a reset word of length 
$$2n^2 -4n +1 -2(n-1) \ln\frac{n}{2}.$$
The paper is organized as follows:
Section \ref{sec2} contains the definitions and elementary results necessary for our pourposes.
In Section \ref{sec3}, we present locally strongly transitive automata.
Reducible sets of states of a locally strongly transitive automaton are studied in Section \ref{sec4}.
In Section \ref{sec5}, we obtain upper bounds for the minimal length of a reset word of a 
locally strongly transitive synchronizing automaton and, more generally, 
for the minimal length of a word taking a reducible set of states of a 
locally strongly transitive automaton into a single state.
The construction of short words of minimal rank is studied in Section \ref{sec6}.
Finally, in Section \ref{sec7} we consider the hybrid \v Cerny-Road coloring problem 
for graphs with a Hamiltonian path.

 Some of the results of this paper were presented in undetailed form at MFCS 2009 \cite{DLT} and at DLT 2010 \cite{DLT10}. 

  \section{Preliminaries}\label{sec2}
    We assume that the reader is familiar with the theory of automata 
and rational series. In this section we shortly recall a vocabulary of few terms  
and we fix the corresponding notation used in the paper.

Let $A$ be a finite alphabet and let $A^{*}$ be the free monoid of words
over the alphabet $A$. The identity of $A^*$ is called the {\em
empty word} and is denoted by $\epsilon$.
 The {\em length} of a word $w\in A^*$ is the integer $|w|$ 
inductively defined by $|\epsilon|= 0$, $|wa| = |w| + 1$, $w\in A^*$, $a\in A$.
For any positive integer $n$, we denote by $A^{<n}$ the set of words of length 
smaller than $n$. 

For any finite set of words,  $W$, we denote  respectively by $L_{W}$ and $\ell _W$ 
    the maximal and minimal lengths of the words of $W$.

A finite automaton is a triple $\A=\langle Q,A,\delta\rangle$ 
where  $Q$ is a finite
set of elements called {\em states} and $\delta$ is a map
$$\delta : Q\times A \longrightarrow Q.$$ 
The map $\delta$ is called the
{\em transition function} of $\cal A$.
The canonical extension
 of the map $\delta$ to the set $Q\times A^{*}$
is still denoted by $\delta$.
 
If $P$ is a subset of $Q$ and $u$ is a word of $A^{*}$,
we denote by $\delta (P, u)$ and $\delta (P, u^{-1})$ the sets:
$$\delta (P, u) = \{\delta (s, u)\mid s\in P\}, \quad \delta (P, u^{-1}) =  \{s\in Q \mid \delta (s, u)\in P\}.$$
 In the sequel,   if no confusion arises, for any set of states $K$ and any $w\in A^*$, we denote by $Kw^{-1}$ the set
$\delta (K, w^{-1})$. 
With any automaton $\A=\langle Q,A,\delta\rangle$, we can associate
a directed multigraph $G = (Q, E)$, where the multiplicity of the edge $(p, q)\in Q\times Q$ 
is given by ${\card (\{a\in A \ \mid \ \delta (p, a) = q\}}$.
If the automaton $\A$ is associated with $G$, we also say that $\A$ is a {\em coloring} of $G$.
An automaton is {\em transitive} if the associated graph is strongly connected. 
 If $n=\card (Q)$, we will say that $\cal A$ is a $n$-state automaton.
 
 The {\em rank} of a word $w$  is the cardinality of the set of states $\delta (Q, w)$.
A {\em synchronizing} or {\em reset} word of $\A$  is any  word $u\in A^*$ of rank $1$.
A
{\em synchronizing }  automaton  is an automaton that has a
reset word. 
The following conjecture has been raised in
\cite{Cerny}.
\medskip\\
{\bf \v Cern\'y Conjecture.} {\em Each synchronizing   $n$-state automaton 
has a reset word of length not larger than $(n-1)^{2}$.}
\medskip
%
 
 Let $\A=\langle Q,A,\delta\rangle$ 
   be  any 
$n$-state automaton. One can associate with $\cal A$ a morphism
$$\fa : A^* \rightarrow \mathbb Q ^{Q\times Q},$$
of the free
monoid $A^*$ in the multiplicative monoid $\mathbb Q^ {Q\times Q}$ of
matrices over the field $\mathbb Q$ of rational numbers,
defined as: for any $u\in A^*$ and for any $s, t\in Q$,  
 $$\fa (u)_{st} \;=\; \left \{ \begin{array}{ll}
1 & \;\mbox{if} \; t = \delta(s,u)\\
0 & \;\mbox{otherwise.}\end{array}
\right.$$
Let us consider a linear order on $Q$ so that  $Q = \{q_1, \dots, q_n\}$. If $K$ is a subset of $Q$, then one can associate with $K$ 
its  {\em characteristic vector}   $\underline{K}\in \Q^Q$ defined as: for every $i=1, \ldots, n$,
$$\underline K_i=\begin{cases}
1&\mbox{if }q_i\in K,\\ 0&\mbox{if }q_i\notin K.
\end{cases}$$
It is easily seen that, for any $S_1, S_2 \subseteq Q$ and $v\in A^*$, one has:
\begin{equation}\label{equation00}\underline{S_1}\varphi_{\A}(v)\underline{S_2} ^t = \card (S_2v^{-1} \cap S_1).
\end{equation}

The following well-known lemma will be used in the sequel. The proof can be found for instance in \cite{Kari} or in \cite{page}. 
\begin{lemma}{(Fundamental lemma)}\label{fl:1}
Let $\varphi : A^* \longrightarrow \Q ^ {Q\times Q}$ be a monoid morphism. Let $\cal V$ be a linear
subspace of dimension $k$ of the vector space $\Q ^ Q$ and let $v\in \Q ^ Q$. 
If $v\varphi (w)\notin \cal V$ for some word $w\in A^*$, then there exists a word $w'\in A^*$ such that
$$v\varphi (w')\notin {\cal V},   \quad \mbox{and} \quad |w'|\leq k.$$
\end{lemma}

  \section{Independent systems of words }\label{sec3}
   In this section, we will present some results that can be obtained by
using some techniques on  independent systems of words. We begin by recalling a definition  introduced in \cite{DLT}. 
 
\begin{definition}\label{sta}
	Let $\A=\langle Q,A,\delta\rangle$  
	be an automaton.
	A set of $k$ words $W=\{w_{1},\ldots, w_{k}\}$  is called {\em independent} if there exist
	$k$ distinct states $q_1,\ldots,q_{k }$ of $\A$  such that, for all $s\in Q$,
	$$
	\{ \delta(s, w_1),\ldots, \delta(s,w_{k })\}=\{q_1,\ldots,q_{k }\}.
	$$
	The set $R=\{q_1,\ldots,q_{k }\}$ will be called the \emph{range} of $W.$
\end{definition}
 An automaton is called {\em locally strongly transitive} if it has an independent set of words. 
 The following  example shows that  local strong transitivity does not imply transitivity.
 
\begin{example}
Consider the $4$-state automaton $\cal A$
 over the alphabet $A=\{a, b\}$ defined by the following graph:
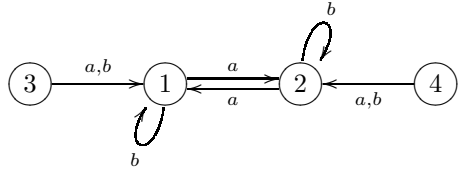
\begin{center}
\automa{
\stato{3}\arco^{a,b}r0&
\stato{1}\ciclo b5\arcoAR aar2&
\stato{2}\ciclo b1&
\stato{4}\arco^{a,b}l0
}
\end{center}
 The automaton $\cal A$ is not transitive. On the other hand, one can easily check that the set $\{a, a^2\}$ is an independent set of
$\A$ with range $R = \{1, 2\}$. 
\end{example}

The following useful properties can be derived from Definition \ref{sta} (see \cite{DLT}, Section 3).
\begin{lemma}\label{lm:pr1bis}
    Let $\cal A$ be an automaton and let $W$ be an independent set of $\A$ with range $R$.
      Then, for every $u\in A^*$, the set 
$uW$ is  an independent set of $\cal A$ with range $R$. \end{lemma}    
 \begin{proposition}\label{pr:1}
      Let $W = \{w_1, \ldots, w_k\}$ be an independent set of a locally strongly transitive automaton $\A=\langle Q,A,\delta\rangle$  
        with range $R$. Then, for every subset $P$ of $R$,
     $$\sum _{i=1}^{k} \;\card (Pw_i^{-1}\cap R) = k \card (P).$$
      \end{proposition}
      
        \begin{proof} 
	Because of Definition \ref{sta}, for every $s\in S$ and $r\in R$, there exists exactly one word $w\in W$ such that
	$s\in  \{r\}w^{-1}.$ This implies that  
	 	the sets $\{r\}w_i^{-1}$, $1\leq i\leq k$, give a partition of $S$.  
	Hence, for any $r\in R$,  one has:  
		\begin{equation}\label{eq8}
	k\;=\; \card (R) \;=\;
	\sum _{i=1} ^{k}\;\card (R\;\cap\; \{r\}w_{i}^{-1}).
	\end{equation}
	Let $P$ be a a subset of $R$. If $P$ is empty then the statement is trivially true. 
	If  $P =\{p_{1},\ldots, p_{m}\}$ is a set of $m\geq 1$ states, then one has:
	$$\sum _{i=1} ^{k} \;\card (R\;\cap\; Pw_{i}^{-1}) \;=\;\sum _{i=1} ^{k} \;\card \left( \bigcup _{j=1}^{m}\; R \;\cap\; \{p_{j}\}w_i^{-1}\right).$$  
	Since $\cal A$ is deterministic, for any pair
	$p_{i}, p_{j}$ of distinct states of $P$ and for every $u\in A^{*}$, one has:
		 $$	\{p_{i}\}u^{-1}\;\cap\; \{p_{j}\}u^{-1} \;=\; \emptyset,$$
		 so that the previous sum can be rewritten as:
 	$$	 \sum _{i=1} ^{k}\ \sum _{j=1} ^{m}\ \card (R \;\cap\; \{p_j\}w_{i}^{-1}).$$  
The latter equation together with  (\ref{eq8}) 
implies that 
	 	$$\sum _{i=1} ^{k}\;\card (Pw_{i}^{-1} \;\cap\; R) \;=\; k\card (P).$$	
 \qed\end{proof} 
 \begin{remark}\label{rmkarturo}
  As an immediate consequence of Proposition \ref{pr:1}, one derives that
  either $\card(Pw_{i}^{-1} \;\cap\; R) \;=\; \card (P)$, for all $i=1,\ldots, k$ or
  $\card(Pw_{j}^{-1} \;\cap\; R) \;>\; \card (P)$, for some $j\in \N$ with $1\leq j \leq k$. 
  \end{remark}
  
 \section{Reducible sets}\label{sec4}
 Let $\A=\langle Q,A,\delta\rangle$  
 be a $n$-state automaton.  
 We say that a set $K$ of states  of $\A$ is \emph{reducible} 
if, for some word $w$, $\delta(K, w)$ is a singleton.

 We now introduce the important notion of stability \cite{ckk}.  
  Given two states $p, q$ of $\A$, we say that the pair $(p, q)$ is {\em stable} if, for all
$u\in A^*$, the set $\{\delta(p, u), \delta(q,u)\}$ is reducible. 
The set $\rho$ of stable pairs is a congruence of the automaton $\A$, which is called \emph{stability relation}. 
It is easily seen that an automaton is synchronizing if and only if the stability relation is 
the universal equivalence.
A set $K\subseteq Q$ is \emph{stable} if for any $p,q\in K$, the pair $(p,q)$ is stable.
Any stable set $K$ is reducible. Thus, even if $\A$ is not synchronizing, one may want to evaluate
the minimal length of a word $w$ such that $\card(\delta (K, w))=1$. 

In the sequel, we assume that $W=\{w_1,\ldots, w_k\}$ is an independent set of $\A$   with range $R$.
We denote by $M$ the maximal cardinality of reducible subsets of $R$.
 The following proposition characterizes maximal  reducible subsets of $R$.
\begin{proposition}\label{prop:clique}
	Let $K$ be a non-empty reducible subset of $R$.
	The following conditions are equivalent:
	\begin{enumerate}
	\item $\card (K)=M$,
	\item for all $w\in W$, $v\in A^{*}$, $\card(K(vw)^{-1}\cap R)\leq\card (K)$,
	\item for all $w\in W$, $v\in A^*$, $\card(K(vw)^{-1}\cap R)=\card (K)$.
	\item $K$ is a maximal  reducible subset of $R$.
	\end{enumerate}
\end{proposition}
\begin{proof}
Implication $1.\Rightarrow2.$ is tivial, since $K(vw)^{-1}\cap R$ is reducible.

Implication $2.\Rightarrow3$ is a straightforward consequence of Remark \ref{rmkarturo},
taking into account that for any $v\in A^*$, the set
$vW$ is independent by Lemma \ref{lm:pr1bis}. 

Now, let us prove implication $3.\Rightarrow 4.$
Let $X$ be a reducible subset of $R$ with $\card(X)=M$.
One has $\delta(X,v)=\{q\}$ and $\delta(q,w)\in K$ for some $v\in A^*$, $q\in Q$, $w\in W$.
Hence, $X\subseteq K(vw)^{-1}\cap R$ so that $\card (K)=\card(K(vw)^{-1}\cap R)\geq M$.
One concludes that $K$ is maximal.

Finally, let us prove implication $4.\Rightarrow 1.$
Let $X$ be a reducible subset of $R$ with $\card(X)=M$.
One has $\delta(K,v)=\{q\}$ and $\delta(q,w)\in X$ for some $v\in A^*$, $q\in Q$, $w\in W$.
Hence, $K\subseteq X(vw)^{-1}\cap R$.
Since $X(vw)^{-1}\cap R$ is reducible, from the maximality of $K$ one obtains $K=X(vw)^{-1}\cap R$.
We have yet proved that  $1.\Rightarrow3.$ It follows that $\card(X(vw)^{-1}\cap R)=\card(X)$,
that is, $\card(K)=M$.\qed
\end{proof}

Our next goal is to evaluate the length of a word $v$ such that $\delta (K, v)$ is a singleton for some maximal
 reducible subset $K$ of $R$. 
\begin{lemma}\label{arturo3:reset}
The condition 
$$\card (K (vw_i)^{-1} \cap R) = \card (K), \ i =1, \ldots, k,$$
holds if and only if the vector $\underline{R}\varphi_{\A} (v)$ is a solution of the system
\begin{equation}\label{sys}
\left \{ 
\begin{array}{l}
  \left (  \underline{Kw_i ^ {-1}}  - \frac{\card (K)}{\card (R)}\underline{Q}   \right ) x = 0
  \\\\
   i =1, \ldots, k.
\end{array}
\right.
\end{equation}
 \end{lemma}
\begin{proof}
By taking into account Equation (\ref{equation00}), 
we obtain\begin{align*}
  &\left (  \underline{Kw_i ^ {-1}}  - \frac{\card (K)}{\card (R)}\underline{Q}   \right ) ( \underline{R}\varphi_{\A}(v))^t 
 = \underline{R}\varphi_{\A}(v)  \left (  \underline{Kw_i ^ {-1}}  - \frac{\card (K)}{\card (R)}\underline{Q}   \right )^t \\
 &\qquad= \card (Kw_i ^{-1}v^{-1} \cap R ) - \frac{\card (K)}{\card (R)}\card (Qv^{-1}\cap R)\\
 &\qquad= \card (K (vw_i)^{-1} \cap R) -  \card (K).
 \end{align*}
The statement then follows from the equality above.
\qed\end{proof}

\begin{lemma}\label{arturo2:resetbis}
Let $A$ be a matrix with $k$ rows. Suppose that no row is null and any column of $A$ has at most $t>0$ non-null entries.
Then $\rk(A) \geq k/t$.  
\end{lemma}
\begin{proof}
 Let $\{c_1, \ldots c_r\}$ be a maximal set  of linearly independent columns of $A$. 
 Hence we have $r = \rk (A)$.    If $rt< k$, there exists an index $i$, with 
$1\leq i \leq k$, such that the entries at position $i$ of $c_1, \ldots, c_r$ are null.
Since all columns of $A$ linearly depend on  $\{c_1, \ldots c_r\}$, 
this implies that the $ith$ row of $A$ is null, contradicting
our assumption.  Thus $rt\geq k$ and the conclusion follows. 
	 \qed\end{proof}

\begin{lemma}\label{arturo2:reset}
Assume that $Kw_i ^{-1}\neq \emptyset$ and $Kw_i ^{-1}\neq Q$, for $1\leq i\leq k$. 
 	The rank of the system  (\ref{sys}) is larger or equal than
	$$ \max \left \{ \frac{\card(R\setminus K)}{\card (K)},\ \frac{\card(K)}{\card(R\setminus K)} \right \}.$$
	 
\end{lemma}
\begin{proof}
Let $C$ be the matrix of the system (\ref{sys}). One has
$$C = A - \frac{\card (K)}{k}U,$$
where
$$A = 
\left(
\begin{array}{c}
   \underline{Kw_1^ {-1}}    \\
  \vdots\\
   \underline{Kw_k^ {-1}}
\end{array}
\right),
$$
and $U$ is the matrix with all entries equal to $1$. 

 Since $W$ is an independent set, any column of $A$ has
exactly $\card (K)$ non-null entries.  By Lemma \ref{arturo2:resetbis}, one has 
$\rk (A) \geq {k}/{\card (K)},$
so that $$\rk (C) \geq \rk(A) - \rk(U) \geq \frac{k}{\card (K)} -1=\frac{\card(R\setminus K)}{\card (K)}.$$
Similarly, one has also that
$$C = A - U + \left(1 - \frac{\card (K)}{k}\right)U.$$
We notice that an entry of the matrix $A - U$ is non-null if and only if the
corresponding entry of $A$ is null. Thus any column of $A - U$ has exactly 
$k - \card (K)$ non-null entries. By Lemma \ref{arturo2:resetbis}, one has 
$\rk (A-U) \geq {k}/{(k-\card (K))},$ so that
$$\rk (C) \geq \rk (A-U) - \rk (U) \geq \frac{{k}}{{k-\card (K)}} - 1 = \frac{\card (K)}{\card(R\setminus K)}.$$
\qed\end{proof}

\begin{lemma}\label{arturo1:reset}
Let $K$ be a non-empty reducible subset of $R$ such that $\card (K) \neq M$. 
 Then there exist a word $v\in A^*$ and a positive integer
$i$ with $1\leq i \leq k$ such that
$$\card (K(vw_i)^{-1}\cap R)>\card(K),$$ and
	\begin{equation}\label{eq:arturo00}
	|v|\leq n - \max \left \{ \frac{\card(R\setminus K)}{\card (K)},\ \frac{\card(K)}{\card(R\setminus K)} \right \}.
	\end{equation}
	 
\end{lemma}
\begin{proof}
Taking into account that, by Lemma \ref{lm:pr1bis},  for any word $v\in A^*$, $\{vw_1, \ldots, vw_k\}$ is 
an independent set with range $R$, in view of
Remark \ref{rmkarturo}, it is sufficient to find a word $v$ such that
$$\card (K(vw_i)^{-1} \cap R) \neq \card (K),$$
for some $i$ with $1\leq i \leq k$. 
Moreover, we may suppose that 
$$Kw_i ^{-1}\neq \emptyset\quad \mbox{and} \quad Kw_i ^{-1}\neq Q,$$
since, otherwise, the inequality above is trivially verified with $v = \epsilon$.

Let $\cal V$ be the space of solutions of the system (\ref{sys}). Since, by hypothesis, $\card (K) \neq M$,
by Proposition \ref{prop:clique} and by Lemma \ref{arturo3:reset},
there exists $v\in A^*$ such that $\underline{R}\varphi_{\A} \notin \cal V$. Moreover, by Lemma \ref{fl:1}, we may suppose that $|v| \leq \dim  \cal V$.  By Lemma \ref{arturo2:reset},
(\ref{eq:arturo00}) holds true.  Hence, by Lemma  \ref{arturo3:reset}, we have $\card (K(vw_i)^{-1} \cap R) \neq \card (K),$
for some $i$ and the claim is  proved.
 \qed\end{proof}

Now we are ready to prove the announced result. 
\begin{proposition}\label{prop:reset}
Let $q\in R$.
	There exist $K\subseteq R$ and $v\in A^*W\cup\{\epsilon\}$ such that
	\[
	\card(K)=M\,,\quad
	|v|\leq(M-1)(L_W+n+1)-k\ln M\,,\quad
	\delta (K, v)=\{q\}\,.
	\]
\end{proposition}
\begin{proof} If $M=1$, the statement is trivially verified by $v=\epsilon$.
Thus we assume $M\geq 2$.
Let $K_0= \{q\}$.
By Lemma \ref{arturo1:reset}, there are subsets 
$K_1, \ldots, K_t$ of $R$, with $t\geq 1$, such that
$$1=\card (K_0)<\card (K_1)< \ldots<\card (K_t)=M,$$
where, for every $i=0, \ldots, t-1,$
$$K_{i+1}= K_i (v_iw_{\gamma_i})^{-1}\cap R,$$
and
$$|v_i|\leq n - \frac{\card(R\setminus K_i)}{\card (K_i)},$$
with $\gamma_i\in \N, 1\leq \gamma_i \leq k$.
By  taking $K = K_t$ and $v = v_{t-1}w_{\gamma_{t-1}}\cdots v_0 w_{\gamma_0}$, we have 
$\card (K)=M$ and $\delta (K, v)=\{q\}$.
 
Moreover, we have
\begin{align*}
|v|& \leq\, \sum _{i=0}	^ {t-1} \left (  n - \frac{\card (R\setminus K_i)}{\card (K_i)} + L_W \right ) 
\leq \, \sum _{j=1}	^ {M-1} \left (  n - \frac{k- j}{j} + L_W \right )\\ 
	&= \, (M-1)(n+L_W+1) -  k \sum _{j=1}	^ {M-1} \, \frac{1}{j} 
\leq\, (M-1)(n+L_W+1) - k \ln M. 	\end{align*}

The statement of the proposition is therefore proved. 
\qed\end{proof}

\section{Some applications}\label{sec5}
 We now present some applications of the results proved in Section \ref{sec4} to stable sets
and to synchronizing automata. As before, let $\A=\langle Q,A,\delta\rangle$  
 be a $n$-state locally strongly transitive automaton where $W=\{w_1,\ldots, w_k\}$ is an independent set of $\A$   with range $R$.
We denote by $M$ the maximal cardinality of reducible subsets of $R$.
We start by proving a useful lemma. 
\begin{lemma}\label{prop:satura}
	Let $K$ be a reducible subset of $R$ of maximal cardinality.
	There is no stable pair in $K\times (R\setminus K)$.
\end{lemma}
\begin{proof}
By contradiction, let $(p,q)\in K\times (R\setminus K)$ be a stable pair.
Then, $\delta (K, v)=\{\delta (p, v)\}$ and $\delta (p, vu) =\delta (q, vu) =s, \ s\in Q$ for some $u,v\in A^*$.
Thus ${\delta(K\cup\{q\}, vu) } =\{s\}$, contradicting the maximality of $K$.\qed
\end{proof}
\begin{proposition}\label{resetC}
	For any stable set $C$ there exists a word $v$ such that
	\[
	\card(\delta (C, v))=1\,,\quad |v|\leq(M-1)(n+L_W+1)-k \ln M + L_W.
	\]
\end{proposition}
\begin{proof}
By Proposition \ref{prop:reset}, there exist $K\subseteq R$ and $u\in A^*$ such that
$\card (K)=M$, $\card(\delta (K, u))=1$, $|u|\leq(M-1)(n+L_W+1)-k \ln M$.
Since $W$ is an independent set with range $R$, there is $w\in W$ such that 
$\delta (C, w)\cap K\neq\emptyset$.
Moreover, $\delta (C,w)$ is a stable subset of $R$.
By Lemma \ref{prop:satura}, one derives $\delta (C, w)\subseteq K$, so that $\card(\delta (C, wu))=\card(\delta (K, u))=1$.
The statement is thus verified for $v=wu$.\qed
\end{proof}
%

The following result   refines the bound of \cite{Acta}. 

\begin{proposition}\label{resetCC}
Any synchronizing $n$-state automaton with an independent set $W$ has a reset word
of length not larger than 
	\[
	 (k -1)(n+L_W+1) - 2k\ln \frac{k+1}{2}+\ell_W\,.
	\]
\end{proposition}
\begin{proof}
In case $M = k$, by following the first part of the proof of Proposition \ref{prop:reset}, one obtains a  word $v$ such that
$\card(\delta (R, v)) = 1$ where
$$v = v_{k-1}w_{\gamma_{k-1}}\cdots v_1 w_{\gamma_1},$$ with
 $w_{\gamma_1}, \ldots, w_{\gamma_{k-1}}\in W$ and

$$|v_i|\leq n - \max \left \{ \frac{\card(R\setminus K_i)}{\card (K_i)}, \ \frac{\card(K_i)}{\card (R \setminus K_i)}\right \}.$$
Therefore one obtains 
\begin{align*}
|v| &\leq (k-1)(n+L_W) - \sum _{j=1} ^ {k-1} \max \left \{ \frac{k - j}{j}, \ \frac{j}{k-j}\right \} 
\\ &=(k-1)(n + L_W +1) - k \sum_{j=1}^{k-1}\ \frac{1}{\min \{j, k-j\}}.
\end{align*}
Let us verify that 

\begin{equation}\label{eq:arturo-perugia}
\sum_{j=1}^{k-1}\ \frac{1}{\min \{j, k-j\}}\geq 2 \ln \frac{k+1}{2}.
\end{equation}

Let $t = \lfloor (k-1)/2 \rfloor$. One easily verifies that 
$\sum _{j=1} ^ t \ 1/j = \sum _{j=k-t} ^ {k-1} \ 1/ (k - j)\geq \ln (t+1),$
and consequently
$$\sum _{j=1} ^ t \ \frac{1}{j} +\sum _{j=k-t} ^ {k-1} \ \frac{1}{k - j}\geq 2\ln (t+1).$$
Thus, if $k$ is odd, then  $\sum_{j=1}^{k-1}\  {1}/{\min \{j, k-j\}} \geq 2\ln (t+1)  = 2\ln ((k+1)/2)$. 
If on the contrary $k$ is even, then 
$\sum_{j=1}^{k-1}\  {1}/{\min \{j, k-j\}} \geq  2\ln (t+1) + 2/k.$ Since 
$\ln ((k+1)/2) - \ln (t+1) = \ln (1 + 1/k) \leq 1/k$, we obtain again (\ref{eq:arturo-perugia}).
 From (\ref{eq:arturo-perugia}) one derives 
$$|v| \leq (k -1)(n + L_W + 1) - 2k \ln \frac{k+1}{2}.$$ 
The claim follows by remarking that, for every $w\in W$, $\card (\delta (Q, wv) ) = 1$.
\qed\end{proof}
In the case of $1$-cluster automata the following corollary  recovers the results of B\'eal et al. \cite{BBP} and Steinberg \cite{Stein} 
with an improved bound.

  \begin{corollary}\label{synch:cluster}
Any synchronizing  $1$-cluster $n$-state automaton  has a reset word
of length 
	\[2n^2 -4n +1 -2(n-1) \ln\frac{n}{2}. 
	\]
\end{corollary}
\begin{proof}
A synchronizing 1-cluster $n$-state automaton has an independent set of the form
$W=\{a^{n-1},\ldots,a^{n-k}\}$, where $a$ is a letter and $k$ is the length of the unique
cycle labelled by a power of $a$.
If $k=n$, then the considered automaton is circular and therefore \cite{Dubuc}
it has a reset word of length $(n-1)^2$.
Since 
$$(n-1)^2\leq 2n^2-4n+1-2(n-1)\ln\frac n2\,,$$
in such a case, the statement is verified.
Thus, we assume $k\leq n-1$.
By {Proposition} \ref{resetCC} and taking into account that $L_W = n-1$ and $\ell_W = n-k$, one has
that there exists a reset word of length not larger than $$2nk-n-k-2k\ln\frac{k+1}2.$$
In order to complete the proof, let us verify that, for $1\leq k< n$,

$$2nk-n-k-2k\ln\frac{k+1}2\leq 2n^2-4n+1-2(n-1)\ln\frac n2\,.$$
This inequality can be rewritten as
\begin{equation}\label{eq:arturonota}
2(n-1)\ln n-2k\ln(k+1)\leq(2n-1+2\ln 2)(n-k-1)\,.
\end{equation}
Using the inequality $\ln x\leq x-1$, one has
\begin{multline*}
	2(n-1)\ln n-2k\ln(k+1)=2k\ln\frac n{k+1}+2(n-k-1)\ln n\\
	\leq 2k\frac{n-k-1}{k+1}+2(n-k-1)(n-1)\leq 2n(n-k-1)\,.
\end{multline*}
This proves (\ref{eq:arturonota}) and the proof is complete.\qed\end{proof}

 \section{Words of minimal rank}\label{sec6}
 We now present some applications of the results proved in Section \ref{sec4} to estimate
the length of a shortest word of minimal rank.
As before, let $\A=\langle Q,A,\delta\rangle$  
 be a $n$-state locally strongly transitive automaton where $W=\{w_1,\ldots, w_k\}$ is an independent set of $\A$   with range $R$.
We denote by $M$ the maximal cardinality of reducible subsets of $R$.
  The following lemma is useful.
\begin{lemma}\label{lem:treset}
	Let $1\leq t\leq \lceil k/M\rceil$.
	There are $t$ pairwise distinct states $q_1,\ldots, q_t \in R$
	 and a word
	$v\in A^*$ such that
	\begin{equation}\label{eq:A}
	\card(q_iv^{-1}\cap R)=M\,,\quad i=1,\ldots,t\,,
	\end{equation}
	\begin{equation}\label{eq:B}
	|v|\leq t(M-1)(L_W+n+1)-tk\ln M\,.
	\end{equation}
\end{lemma}
\begin{proof}
	We proceed by induction on $t$. If $t=1$, the claim follows from Proposition \ref{prop:reset}. 
	
	Let us prove the inductive step. 
	For the sake of induction, suppose we have found 
	pairwise distinct states $q_1,\ldots, q_{t-1} \in R$ and a word
	$v'\in A^*$ such that
	\[
	\card(q_iv'^{-1}\cap R)=M\,,\quad i=1,\ldots,t-1\,,
	\]\[
		|v'|\leq (t-1)(M-1)(L_W+n+1)-(t-1)k\ln M\,.
	\]
	Since $(t-1)M<k$, there exists $q\in R\setminus\bigcup_{i=1}^tq_iv'^{-1}$.
	By Proposition \ref{prop:reset}, there exist $K\subseteq R$ and $u\in A^*W\cup\{\epsilon\}$ such that
	\[
	\card(K)=M\,,\quad
	|u|\leq(M-1)(L_W+n+1)-k\ln M\,,\quad
	\delta (K,u)=\{q\}\,.
	\]
	Set $q_t=\delta(q,v')$ and $v=uv'$.
	Clearly, $v$ satisfies (\ref{eq:B}).
	Taking into account Proposition \ref{prop:clique}, one verifies that also (\ref{eq:A}) is satisfied,
	concluding the proof.\qed
\end{proof}

\begin{proposition}\label{minrkArturo}
	The minimal rank of the words of $\A$ is $t=k/M$.
	Moreover, there is a word $u$ of rank $t$ with
	\begin{equation}
		|u|\leq\ell_W+(k-t)(L_W+n+1)-tk\ln\frac kt\,.
	\label{eq:C}
	\end{equation}
\end{proposition}
\begin{proof}
	Applying the previous lemma in the case $t=\lceil k/M\rceil$,
	one finds a word $v$ satisfying (\ref{eq:B}) such that $R$ may be partitioned by the sets
	$q_iv^{-1}$, $i=1,\ldots,t$, of cardinality $M$.
	Hence, $k=tM$.
	
	Let us verify that $t$ is the minimal rank of the words of $\A$.
	Indeed, let $u'$ be a word of rank smaller than $t$.
	Then one has $\delta(q_i,u')=\delta(q_j,u')=q$ for some $i,j$, $1\leq i<j\leq t$, $q\in Q$.
	It follows that $(q_iu^{-1}\cup q_ju^{-1})\cap R$ is reducible, which contradicts the fact
	that this set has cardinality $2M$.
	On the other side, if $u=wv$ with $w\in W$, then
	$\delta(Q,u)\subseteq\delta(R,v)=\{q_1,\ldots,q_t\}$ so that $u$ has rank $t$.
	
	To complete the proof, it is sufficient to check that, choosing $w\in W$
	of minimal length,  (\ref{eq:C}) holds true.\qed	
\end{proof}

As an immediate consequence of {Proposition} \ref{resetC} and {Proposition} \ref{minrkArturo}, we
obtain the following three corollaries.

   \begin{corollary}\label{prop:cliquesize}
   Let $t$ be the minimal rank of $\A$. Then,
	for any stable set $C$ there exists a word $v$ such that
	\[
	\card(\delta (C, v))=1\,,\quad |v| \leq \left( \frac{k}{t} - 1\right)(n+L_W+1) - k\ln \frac{k}{t} + L_W \,.
	\]
\end{corollary}
 
   \begin{corollary}\label{cor:cluster}
  Let $t$ be the minimal rank of   a 1-cluster $n$-state automaton.
  	Then,
	for any stable set $C$ there exists a word $v$ such that
	\[
	\card(\delta (C, v))=1\,,\quad |v| \leq  \frac{2nk}{t} - n - 1  - k\ln \frac{k}{t}.
	\]

\end{corollary}


  \begin{corollary}\label{cor:clusterbis}
  Let $\A$ be  a 1-cluster $n$-state automaton which is not synchronizing. 
  	Then,
	for any stable set $C$ there exists a word $v$ such that
	\[
	\card(\delta (C, v))=1\,,\quad |v| \leq  n^2 - n - 1 - n\ln \frac{n}{2}.
	\]

\end{corollary}

\begin{proof}
By the previous corollary, it is sufficient verify that 
$$ \frac{2nk}{t}  - k\ln \frac{k}{t} \leq n^2 - n\ln \frac{n}{2}.$$
Indeed, one has
\begin{align*}
n\ln \frac{n}{2} -  k\ln \frac{k}{t} &= (n-k) \ln \frac{n}{2} + k \ln \frac{n}{k} + k \ln \frac{t}{2}\\
&\leq (n-k) \left ( \frac{n}{2} -1 \right) + k \left (  \frac{n}{k }  -1 \right) + k \left (  \frac{t}{2 }-1 \right)\\
& \leq n (n - k) + \frac{nk}{t}(t-2) = n^2 - \frac{2nk}{t}.
\end{align*}
The conclusion follows. \qed
\end{proof}
 
\section{The Hybrid \v Cern\'y-Road coloring problem}\label{sec7}
 In the sequel, with the word {\em graph}, we will term a
 finite, directed multigraph
 with all vertices of the same outdegree.
 A graph is {\em aperiodic} if the greatest common divisor of the lengths of all cycles
 of the graph is $1$. 
  A graph is called an {\em AGW-graph} if it is aperiodic and strongly connected. 
  A synchronizing automaton which is a coloring of a graph $G$ will be called a
  {\em synchronizing coloring} of $G$.
 The {\em Road coloring problem} asks for the existence of a synchronizing coloring for every  AGW-graph.
 This problem was formulated in the context
 of Symbolic Dynamics by Adler, Goodwyn and  Weiss and it is explicitly stated in \cite{AGW}. 
 In 2007, Trahtman has positively solved this problem  \cite{trah}.  
Recently Volkov has raised the following problem \cite{V} (see also \cite{AGV}).
\medskip\\
{\bf Hybrid \v Cern\'y--Road coloring problem.} {\em Let $G$ be an  AGW-graph. What is the minimum length of a reset word for a 
synchronizing coloring of $G$?
}
 
\subsection{Relabeling}\label{sec:recoloring}
In order to prove our main theorem, we need to recall some basic results concerning 
colorings of graphs.
  Let  $\A=\langle Q,A,\delta\rangle$ be an automaton. A map $\delta' : Q\times A \longrightarrow Q$ 
  is a {\em relabeling} of $\A$  if, for
each $q\in Q$,  there exists a permutation $\pi_q$ of $A$ such that
\[\delta'(q,a)=\delta(q,\pi_q(a)), \ a\in A\,.\]
It is clear that $\delta'$ is a relabeling of $\A$ if and only if the automata $\A$ and $\A' =\langle Q,A,\delta'\rangle$
are associated with the same graph.

  Let  $\A=\langle Q,A,\delta\rangle$ be an automaton,
$\alpha$ be a congruence on $Q$
and $\delta'$ be a relabeling of $\A$.
According to \cite{ckk}, $\delta'$ \emph{respects} $\alpha$ if for
each congruence class $C$ there exists a permutation $\pi_C$ of $A$ such that
\[\delta'(q,a)=\delta(q,\pi_C(a))\,,\quad q\in C\,,\ a\in A\,.\]
In such a case, for all $v\in A^*$ there is a word $u\in A^*$ such that
$|u|=|v|$ and
$\delta'(q,u)=\delta(q,v)$ for all $q\in C$.

As $\alpha$ is a congruence, we may consider the quotient automaton $\A/\alpha$.
Any relabeling $\widehat\delta$ of $\A/\alpha$ induces a relabeling 
$\delta'$ of $\A$ which respects $\alpha$.
This means that
\begin{enumerate}
	\item $\alpha$ is a congruence of $\A'=\langle Q,A,\delta'\rangle$
	and $\A'/\alpha=\langle Q/\alpha,A,\widehat\delta\rangle$,
	\item 
	for all $\alpha$-class $C$ and all $v\in A^*$, there exists $u\in A^*$ such that
	$|v|=|u|$ and $\delta'(C,u)=\delta(C,v)$.
\end{enumerate}

We end this section by recalling the following important result proven in \cite{ckk}.
\begin{proposition}{}\label{recoloring-kar}
Let $\rho$ be the stability congruence of an automaton $\A$ associated with an AGW-graph $G$.
Then the graph $G'$ associated with the quotient automaton  $\A /\rho$ is  an AGW-graph. Moreover, if $G'$ has
a synchronizing coloring, then $G$ has a synchronizing coloring as well.
\end{proposition}

\subsection{Hamiltonian paths}\label{sec:paths}
In this section we give a partial answer to the Hybrid \v Cern\'y--Road coloring problem. Precisely we prove that an
AGW-graph of $n$ vertices with a  Hamiltonian path admits a synchronizing coloring with a reset  word of length not
larger that $2n^2 -4n +1 -2(n-1) \ln ({n}/{2}).$  In order to prove this result, 
we need to establish some properties concerning automata with a monochromatic Hamiltonian path.
 
 Let  $a$ be a letter. 
The graph of $a$-transitions of an automaton $\A$ consists of 
disjoint cycles and trees with root on the cycles.
The \emph{level} of a vertex in such a graph is its height in the tree to which it belongs. 
The following proposition was implicitly proved in \cite[Theorem 3]{trah}.

\begin{proposition}\label{prop:trah}
	If in the graph of $a$-transitions of a transitive automaton $\A$
	all the vertices of maximal positive level belong to the same tree, then $\A$ has a stable pair.
\end{proposition}

As an application of the previous proposition, we obtain the following. 
\begin{proposition}\label{prop:hamilton}
If an AGW-graph $G$ with at least 2 vertices has a Hamiltonian path, 
then there is a coloring of $G$ with a stable pair and a monochromatic Hamiltonian path.
\end{proposition}
\begin{proof}  Let $G$ be an AGW-graph with $n \geq 2$ vertices. 
Let us show that one can find in $G$ a Hamiltonian path $(q_0,q_1,\ldots,q_{n-1})$
and an edge $(q_{n-1},q)$ with $q\neq q_0$ (see fig.).
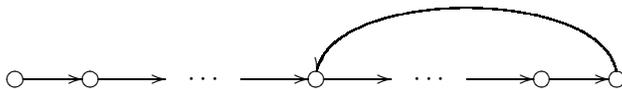
\[\xy 0;<10mm,0mm>:%
(0,0)*\cir<3pt>{}%
\ar(1,0)*\cir<3pt>{}="1"%
\ar"1";(2,0)%
\ar(3,0);(4,0)*\cir<3pt>{}="q"%
\ar"q";(5,0)%
\ar(6,0);(7,0)*\cir<3pt>{}="1"%
\ar"1";(8,0)*\cir<3pt>{}="1"%
\ar@(u,u)"1";"q"%
\POS(2.5,0)*{\cdots}\POS(5.5,0)*{\cdots};%
\endxy\]
Indeed, if $G$ has no Hamiltonian cycle, it is sufficient to take 
a Hamiltonian path $(q_0,q_1,\ldots,q_{n-1})$ and any edge outgoing from $q_{n-1}$:
such an edge exists because $G$ has positive constant outdegree.

On the contrary, suppose that $G$ has a Hamiltonian cycle
$(q_0,q_1,\ldots,q_{n-1},q_0)$.
Since $G$ is aperiodic, there is an edge $(p,q)$ of $G$ which does not belong to the cycle.
We may assume, with no loss of generality, $p=q_{n-1}$, so that $q\neq q_0$.
Thus, $(q_0,q_1,\ldots,q_{n-1})$ is a Hamiltonian path and $(q_{n-1},q)$ is an edge of $G$.

Choose a coloring $\A$ of $G$ where the edges of the path $(q_0,q_1,\ldots,q_{n-1},q)$ 
are labeled by the same letter $a$.
In such a way, there is a monochromatic Hamiltonian path.
Moreover, the graph of $a$-transitions has a unique tree, so that, by Proposition \ref{prop:trah},
$\A$ has a stable pair.\qed
\end{proof}

\begin{lemma}\label{lem:quotient}
If an automaton $\A$ has a monochromatic Hamiltonian path, then any quotient automaton of $\A$ 
has the same property.
\end{lemma}
\begin{proof}
With no loss of generality, we may reduce ourselves to the case that $\A$ is a 1-letter automaton.
Now, a 1-letter automaton has a Hamiltonian path if and only if it has a state $q$ from which
all states are accessible.
The conclusion follows from the fact that the latter property is inherited by 
the quotient automaton.\qed
\end{proof}

We are ready to prove our main result. We denote by $f$ the real function 
$$f(x) =  2x^2 -4x +1 -2(x-1) \ln\frac{x}{2}.$$
One easily verifies that, for $x\geq 2$, one has $f'(x) \geq x$. In particular,
$f$ is strictly increasing. 

\begin{theorem}\label{prop:resethamilton}
Let $G$ be an AGW-graph with $n>1$ vertices.
If $G$ has a Hamiltonian path, then there is a synchronizing coloring  of $G$
with a reset word $w$ of length 
\begin{equation}\label{eq:lung}
	|w| \leq 2n^2 -4n +1 -2(n-1) \ln\frac{n}{2}. 
\end{equation}
\end{theorem}

 \begin{proof} The proof is by induction on the number $n$ of the vertices of $G$.
 
Let $n=2$. Since $G$ is aperiodic, $G$ has an edge $(q,q)$ which immediatly implies the statement.
Suppose $n\geq 3$.
By Proposition \ref{prop:hamilton}, 
among the colorings of $G$, there is 
an automaton $\A =\langle Q, A,\delta\rangle$
with a stable pair and a monochromatic Hamiltonian path. 
 In particular,  $\A$ is a transitive 1-cluster automaton.
If $\A$ is synchronizing, then the statement follows from Corollary \ref{synch:cluster}. 
Thus, we assume that  $\A$ is not synchronizing.
 
Let $\rho$ be the stability congruence of $\A$, $k$ be its index and $G_\rho$ be the graph of $\A/\rho$ respectively.
Since $\A$ is not synchronizing, one has $k>1$. 
By Proposition \ref{recoloring-kar}, $G_\rho$ is an AGW-graph with $k$ vertices and $k<n$.
Moreover, by Lemma \ref{lem:quotient}, $G_\rho$ has a Hamiltonian path.
By the induction hypothesis, we may assume that there is a relabeling $\widehat\delta$ of $\A/\rho$ 
such that the automaton $\widehat \A = \langle Q/\rho, A, \widehat \delta\rangle$ has a reset word 
$u$ such that 
\[|u|\leq f(k).
	\]
As viewed in Section \ref{sec:recoloring}, $\widehat\delta$ induces a relabeling  
$\delta'$ of $\A$ which respects $\rho$.
Moreover, since $u$ is a reset word of $\widehat\A$, $C=\delta'(Q,u)$ is a stable set of $\A$.
 
First, we consider the case $n\geq 2k$.
By Corollary \ref{cor:clusterbis}, there is a word $v$ such that $\card(\delta(C,v))=1$ and  
$|v|\leq n^2 - n\ln{n}/{2} - n - 1.$ 
Since $\delta'$ respects $\rho$,
there is a word $v'$ such that 
$|v'|=|v|$ and $\delta'(C,v')=\delta(C,v)$.
Set $w=uv'$.
Then $\delta'(Q,w)=\delta'(Q,uv')=\delta'(C,v')=\delta(C,v)$ 
is reduced to a singleton.
Hence, $w$ is a reset word of $\A'= \langle Q, A, \delta'\rangle$ and
\[
|w|\leq f(k) + n^2 - n\ln\frac{n}{2}  - n - 1.
\]
Since $f$ is increasing and $k\leq n/2$, one has 
\begin{align*}
f(n) - |w| & \geq f(n) - f\left(\frac{n}{2}\right) -  \left(n^2 - n\ln\frac{n}{2}  - n - 1\right) \\
& =  \frac{1}{2}n^2 - (1 + \ln 2)n + 1 + \ln 4 > 0.
\end{align*}
Hence (\ref{eq:lung}) holds true.
 
Now, we consider the case $n<2k$.
In such a case, there is a $\rho$-class $K$ of cardinality 1.
Moreover, by the transitivity of $\widehat\A$, there is a word $v\in A^*$ such that
$\delta'(C,v)=K$ and $|v|\leq k-1$.
Hence, $w=uv$ is a reset word of $\A'$ of length
\[|w|\leq f(k) + k - 1. 
\]
Since $f'(x) \geq x$, by Lagrange Theorem, one has $f(n) - f(k) \geq (n-k) k\geq k$. It follows that
$|w| \leq f(n) - 1$. 
This concludes the proof.\qed
\end{proof}

We close the paper with the following remark.

\begin{remark}
It was already observed  in \cite{DLT} that a bound on synchronizing $1$-cluster 
automata with prime length cycle leads to bounds for the  Hybrid \v Cern\'y--Road coloring problem.
More precisely, by a result of O' Brien  \cite{OB}, every aperiodic graph of $n$ vertices, without multiple edges,  having  
a simple cycle $C$ of prime length $p < n$, admits a synchronizing coloring of $G$ such that $C$ is the unique   cycle
  labelled by a power of a given letter $a$. Then, by Corollary \ref{synch:cluster}, such coloring has a reset word of
  length $2n^2 -4n +1 -2(n-1) \ln({n}/{2})$. Recently this upper bound has been lowered to $(n-1)^2$ in \cite{Steinrc}. 
  
   \end{remark}

\end{document}